\documentclass[preprint,12pt]{elsarticle}

\usepackage{amssymb}
\usepackage{amsmath}

\journal{Information Processing Letters}

\usepackage{hyperref}
\usepackage{amsthm}
\newtheorem{theorem}{Theorem}[section]
\newtheorem{definition}[theorem]{Definition}
\newtheorem{lemma}[theorem]{Lemma}

\DeclareMathOperator{\Lap}{Lap}

\usepackage[dvipsnames]{xcolor}
\usepackage{tikz}
\usetikzlibrary{matrix, positioning}
\begin{document}

\begin{frontmatter}



\title{Private Graph Colouring with Limited Defectiveness} 


\affiliation[DTU]{organization={Technical University of Denmark},
            addressline={Richard Petersens Plads},
            city={Kgs. Lyngby},
            postcode={2800},
            country={Denmark}}

\affiliation[SDU]{organization={University of Southern Denmark},
            addressline={Campusvej 55},
            city={Odense},
            postcode={5000},
            country={Denmark}}

\author[DTU]{Aleksander B. G. Christiansen}
\ead{abgch@dtu.dk}
            
\author[DTU]{Eva Rotenberg}
\ead{erot@dtu.dk}

\author[SDU]{Teresa Anna Steiner}
\ead{steiner@imada.sdu.dk}

\author[DTU]{Juliette Vlieghe\corref{cor1}}
\ead{jmvvl@dtu.dk}
\cortext[cor1]{Corresponding author}

\begin{abstract}
Differential privacy is the gold standard for privacy preserving data analysis, which is crucial in a wide range of disciplines. Vertex colouring is one of the most fundamental graph problems. In this paper, we study the vertex colouring problem in the differentially private setting. 

In this paper, we consider edge-differential privacy. 
To satisfy non-trivial privacy guarantees under this notion of privacy, a colouring algorithm needs to be defective: a colouring is $d$-defective if a vertex can share a colour with at most $d$ of its neighbours. Without defectiveness, any edge-differentially private colouring algorithm needs to assign $n$ different colours to the $n$ different vertices. We show the following lower bound for the defectiveness: any $\epsilon$-egde differentially private algorithm that returns a $c$-vertex colouring of a graph of maximum degree $\Delta > 0$ with high probability must have defectiveness at least $d \in \Omega \left(\frac{\log n}{\epsilon + \log c}\right)$.

We complement our lower bound by presenting an $\epsilon$-differentially private algorithm for $O\left(\frac{\Delta}{\log n}+\frac{1}{\epsilon}\right)$-colouring a graph with defectiveness at most $O\left(\log n\right)$.
\end{abstract}




\end{frontmatter}



\paragraph{Acknowledgements}
This work was supported by the VILLUM Foundation grants VIL37507 ``Efficient Recomputations for Changeful Problems'' and VIL51463 ``Differential Privacy for String Algorithms and Data Structures''.

\section{Introduction and related work}

Graph colouring is a family of fundamental problems with many applications in computer science, including scheduling, routing, register allocation, visualisation, network analysis, and clustering problems.
The vertex colouring problem is the following: one wants to assign each vertex a colour such that no two neighbours have the same colour, that is, each colour class is an independent set.
In most applications of colouring, the edges of the graph represent some kind of conflict: the two adjacent vertices cannot be scheduled at the same time, cannot be visualised with the same colour, or cannot be in the same cluster or team. 
In this paper, we consider the hypothesis that a conflict is sensitive information, and ask whether it is possible to somehow approximately colour the vertices, without revealing their exact conflicts. Is it possible to group the vertices of the graph so that each vertex is only grouped with a limited number of vertices with whom it has a conflict? 
This notion of `approximately' colouring the graph is that of defective colouring:
a colouring is $(c, d)$-defective if it uses a colour palette of size $c$ and any vertex can share a colour with at most $d$ of its neighbours. In other words, we want to partition the set of vertices $V(G)$ into sets $V_1, \dots, V_c$ such that for every $i\in [c]$, each vertex in $V_i$ has at most $d$ neighbours in $V_i$. 
There has been a lot of interesting work on the defective colouring of graphs
\cite{D1987,BIMOR10,CowenGoddardJesurum97,
eaton1999defective,erdHos1967decomposition,
FRICK199345,Hendrey_Wood_2019,JING2022112637,
vskrekovski1999list,OssonaDeMendez2019,
HJWD2028,wood2018defective}, also in distributed graph algorithms~\cite{BarenboimElkin09,Kuhn09, BarenboimElkin11,FuchsKuhn23}.

There is a large body of work studying graphs under differential privacy, including estimating subgraph counts  \cite{DBLP:conf/innovations/BlockiBDS13,DBLP:conf/sigmod/ChenZ13,DBLP:journals/tods/KarwaRSY14,DBLP:conf/sigmod/ZhangCPSX15,DBLP:conf/ccs/SunXKYQWY19,DBLP:conf/uss/ImolaMC21,DBLP:conf/esa/FichtenbergerHO21,DBLP:conf/uss/ImolaMC22,DBLP:conf/ccs/ImolaMC22,DBLP:conf/icalp/EdenLRS23}, the degree distribution of the graph \cite{DBLP:conf/icdm/HayLMJ09,DBLP:conf/sigmod/DayLL16,DBLP:conf/focs/RaskhodnikovaS16,DBLP:journals/compsec/ZhangNF21,DBLP:conf/esa/FichtenbergerHO21}, and the densest subgraph  \cite{DBLP:conf/icml/NguyenV21, DBLP:conf/aistats/0001HS22,DBLP:conf/focs/DhulipalaLRSSY22}, as well as approximating the minimum spanning tree and computing clusterings \cite{DBLP:conf/stoc/NissimRS07,DBLP:conf/pods/NissimSV16,DBLP:conf/pods/HuangL18,DBLP:conf/nips/StemmerK18,DBLP:conf/iccns/LinGHZL19,DBLP:conf/icml/CohenKMST21,DBLP:conf/esa/FichtenbergerHO21,DBLP:conf/kdd/Cohen-AddadELMM22,DBLP:conf/focs/DhulipalaLRSSY22} and cuts and shortest paths \cite{DBLP:conf/soda/GuptaLMRT10,DBLP:conf/pods/Sealfon16,DBLP:conf/nips/AroraU19,DBLP:conf/soda/EliasKKL20,DBLP:conf/pods/Stausholm21,DBLP:conf/soda/ChenG0MNNX23,DBLP:conf/wads/DengGUW23}.

When studying differential privacy, the hope is to give algorithms that disclose some information or analysis of the whole data set (in this case, the graph), while keeping the individual data points (in this case, the edges) private.  
A graph colouring, and especially one with few colours, gives significant insight into the graph structure, which should intuitively reveal much of the private information.
We formalise this intuition by giving a lower bound on the defectiveness of any differentially private colouring.

\subsection{Roadmap}
In section \ref{sec:lower_bound}, we prove the following lower bound: any algorithm that $c$-vertex colour with high probability
graphs of $n$ nodes of degree at most $\Delta$ must have a defectiveness at least $d \in \Omega \left(\frac{\log n}{\log c+\log \Delta}\right)$. The proof builds on the following idea: Given a valid $(c, d)$-defective colouring on a graph $G$, if there exist $d+2$ nodes of the same colour, then there is a $(d+1)$-neighbouring graph to $G$ on which this colouring is invalid: namely, connecting these $d+2$ nodes into a star. This can be used to upper-bound the probability of any colouring producing $d+2$ nodes of the same colour. However, if $c<\frac{n}{d+1}$, then for any valid $(c,d)$-defective colouring there \emph{must} exist $d+2$ nodes of the same colour. Combining these two observations yields our lower bound. 

In section \ref{sec:algorithms}, we propose a simple colouring algorithm. With Chernoff bounds and a random colouring, we can get defectiveness $\Omega(\log n)$, so we minimise the number of colours used, which results in a $\epsilon$-edge differentially private $\left(O\left(\frac{\Delta}{\log n}+\frac{1}{\epsilon}\right), O\left(\log n\right)\right)$ colouring. 

\subsection{Definitions and notations}

We recall some basic definitions and lemmas regarding differential privacy in graphs.

\begin{definition}[Edge-neighbouring graphs~\cite{DBLP:conf/icdm/HayLMJ09}]
    Let $n\in \mathbb{N}$. Let $G = (V, E)$ and $G' = (V, E')$ be two graphs on $n$ nodes. We say that $G$ and $G'$ are $k$-edge-neighbouring if the cardinality of the symmetric difference of $E$ and $E'$ is $k$. In particular, we say that $G$ is an edge-neighbouring graph of $G'$ if they differ by exactly one edge. 
\end{definition}
In this work, we only consider edge-neighbouring (as opposed to \emph{node-neighbouring}, see~\cite{DBLP:conf/icdm/HayLMJ09}). We will thus refer to edge-neighbouring graphs as simply \emph{neighbouring} graphs.

\begin{definition}[Differential privacy \cite{DBLP:conf/tcc/DworkMNS06}]
    A randomised algorithm $A$ is $(\epsilon, \delta)$-differentially private if for all neighbouring inputs $x, y$ and for all possible sets of outputs $S\subseteq \mathrm{range}(A)$:
    \[\Pr(A(x)\in S) \leq e^\epsilon \Pr(A(y)\in S) + \delta.\]
    A randomised algorithm is $\epsilon$-differentially private if it is $(\epsilon,0)$-differentially private.
\end{definition}


In this work, we focus on $\epsilon$-edge differential privacy.
\begin{definition}[Edge differential privacy]
 A randomised algorithm $A$ on $n$-node graphs is $\epsilon$-edge differentially private if for all edge-neighbouring graphs $G$ and $G'$ and for all possible sets of outputs $S\subseteq \mathrm{range}(A)$:
    \[\Pr(A(G)\in S) \leq e^\epsilon \Pr(A(G')\in S).\]
\end{definition}

The next Lemma follows from the definition of differential privacy.
\begin{lemma}[Group privacy \cite{DBLP:conf/tcc/DworkMNS06}]
    Let $G$ and $G'$ be $k$-edge neighboring. Let $A$ be an $(\epsilon, 0)$-edge differentially private on $n$-node graphs. Then for any possible set of outputs $S\subseteq \mathrm{range}(A)$ we have
    \[\Pr(A(G)\in S) \leq e^{k\epsilon} \Pr(A(G')\in S).\]
\end{lemma}

Next, we recall the Laplace mechanism and its properties.
\begin{definition}[Laplace distribution]
    The Laplace Distribution (centered at 0) with scale b is the distribution with probability density function:
    \[\Lap(x|b) = \frac{1}{2b}\exp\left(-\frac{|x|}{b}\right)\]
\end{definition}

\begin{lemma}[Laplace Tailbound]\label{lem:laplace_tail}If $Y \sim \Lap(b)$, then: $P(|Y| \geq t\cdot b) = e^{-t}$.\end{lemma}

In the following, we denote the data universe by $\chi$. For example, $\chi$ can be the set of all graphs with $n$ vertices.

\begin{definition}[$L_1$-sensitivity]
Let $f:\chi\rightarrow \mathbb{R}^k$. The \emph{$L_1$-sensitivity of $f$} is given by
$\max_{x,y\textnormal{ neighbouring}}||f(x)-f(y)||_1$.
\end{definition}

\begin{definition}[Laplace Mechanism]\label{lem:laplace_mech}
    Given $f:\chi\rightarrow \mathbb{R}^k$ with $L_1$-sensitivity $S_1$, the Laplace mechanism is defined as:
    \[M_L(x, f, \epsilon) = f(x) + (Y_1, ..., Y_k),\]
    where $Y_i$ are i.i.d random variables drawn from $\Lap(S_1/\epsilon)$.
\end{definition}

\begin{lemma}[\cite{DBLP:conf/tcc/DworkMNS06}] The Laplace mechanism preserves $\epsilon$-differential privacy.\end{lemma}

We will use the following Chernoff bound:

\begin{lemma}[Additive Chernoff bound]\label{lem:chernoff}
    Let $X_1,\dots, X_m$ be independent random variables s.t. $0\leq X_i \leq 1$. Let $S$ denote their sum and $\mu = \mathbb{E}(S)$. 
    Then for any $\eta \geq 0$:
    \[P(S \geq (1+\eta)\mu) \leq e^{-\frac{\eta^2\mu}{2+\eta}}\]
    Therefore for any $0 \leq \eta \leq 1$:
    \[P(S \geq (1+\eta)\mu) \leq e^{-\frac{\eta^2\mu}{3}}\]
\end{lemma}

\section{A lower bound on the defectiveness of private colouring with high probability}
\label{sec:lower_bound}
In the following, we state our main observation, which is a lower bound on the defectiveness $d$ for any colouring with less than $n$ colours.
\begin{theorem}\label{thm:lower_bound}
   Let $c,d,\Delta, n\in \mathbb{N}$ with $2\leq c<n$ and $d< \Delta$, and let $\epsilon>0$. Consider an $\epsilon$-differentially private algorithm that with high probability returns a $(c, d)$-defective colouring of any $n$-node graph with maximum degree at most $\Delta$. It must be that:
   \[d \in \Omega\left( \frac{\log n}{\log c +\epsilon}\right)\]
\end{theorem}
\begin{proof} 
If $c(d+1)\geq n$, then $\frac{\log n}{\log c + \log (d+1)} \in O(1)$.
Note that $d$ cannot be 0 as otherwise, $c\geq n$ in contradiction to our assumption, so $d\geq 1$ and therefore it holds that $d\in\Omega\left( \frac{\log n}{\log c + \log (d+1)}\right)$. Either $c \geq d$ and the theorem holds, or $d > c$ and therefore $d \in \Omega(\sqrt{n}) \subseteq \Omega\left( \frac{\log n}{\log c + \epsilon}\right)$.

    In the following, we consider the case where $c(d+1) < n$. Let $\alpha >0$ be a constant and assume the algorithm returns a valid $(c,d)$-defective colouring with probability at least $1-\frac{1}{n^{\alpha}}$ for any graph of degree at most $\Delta$.
    Define $n_0 = c(d+1) + 1$.  Any valid $(c, d)$-defective colouring on a graph $G$ has to have the following property: In any subset $V_0$ of $n_0$ nodes in $G$, there has to exist a colour such that at least $d + 2$ nodes have that colour, by the pigeonhole principle. Since our algorithm returns a valid $(c,d)$-defective colouring with high probability on any graph $G$ with maximum degree $\Delta$, it has to fulfil for any subgraph $V_0$ of $n_0$ nodes:
    \begin{align}P(\nexists (d+2) \text{ nodes in }V_0\text{ with the same colour in }G) \leq \frac{1}{n^\alpha}.\label{eq:samecol}\end{align}

    Consider the empty graph $G$ with $n$ nodes. Pick any subset $V_0$ with $n_0$ nodes. Pick any set $U\subseteq V_0$ with $d+2$ nodes. We can find a $d+1$-neighbouring graph $G'$ such that any valid $(c,d)$-defective colouring needs at least 2 colours within $U$: We can add $(d+1)$ edges to form a star with $d+1$ leaves. Since the center can have at most $d$ neighbours of the same colour, we need at least two colours to colour the star. Since our algorithm returns  a valid $(c, d)$-defective colouring with high probability on $G'$, it has to fulfill:

    \begin{align*}
        P(\text{all nodes in $U$ have the same colour in  $G'$})&\leq \frac{1}{n^\alpha}\end{align*}
and by group privacy,
        \begin{align*}
        P(\text{all nodes in $U$ have the same colour in } G)&\leq e^{\epsilon(d+1)}\frac{1}{n^\alpha}.
    \end{align*}
    
    There are $\binom{n_0}{d+2}$ ways of choosing $U\subseteq V_0$, therefore, by the union bound:
    \begin{align*}
        P(\exists (d+2) \text{ nodes in }V_0\text{ with the same colour in }G) 
        \leq &\binom{n_0}{d+2}e^{\epsilon(d+1)}\frac{1}{n^\alpha}\end{align*}
        Combined with (\ref{eq:samecol}), this gives:
        \begin{align*}
        \frac{1}{n^\alpha} \geq P(\nexists (d+2) \text{ nodes in }V_0\text{ with the same colour in }G) 
        \geq 1- &\binom{n_0}{d+2}e^{\epsilon(d+1)}\frac{1}{n^\alpha}\end{align*}
        and further:
        \begin{align*}
        \frac{1}{n^\alpha} \geq 1 - &\binom{n_0}{d+2}e^{\epsilon(d+1)}\frac{1}{n^\alpha}\\
        \geq 1 - &\binom{n_0}{d+2}e^{\epsilon(d+2)}\frac{1}{n^\alpha}.\end{align*}
    We multiply with $n^{\alpha}$ and use the following upper bound on binomial coefficients: $\binom{n_0}{d+2}\leq \left(\frac{n_0 e}{d+2}\right)^{d+2}$. This gives
    \begin{align*}
    n^\alpha &\leq 1 + \binom{n_0}{d+2}e^{\epsilon(d+2)}\\
        &\leq 1 + \left(\frac{n_0e^{\epsilon + 1}}{d+2}\right)^{d+2}\\
       &\in O\left(\left(ce^{\epsilon + 1}\right)^{d+2}\right)
    \end{align*}
    
    Taking the logarithm, we get
    \begin{align*}
        \alpha\log n 
        &\in O\left( (d+2)(\log(c) + \epsilon+1 \right)\\
        &\in O\left( d(\log(c) + \epsilon \right)\\
    \end{align*}
    Therefore:
    \begin{align*}
        d \in \Omega\left(\frac{\log n}{\log c + \epsilon}\right) 
    \end{align*}
which concludes the proof.\end{proof}

Note that if $c \geq n$ we could use a different colour for each vertex, and if $d\geq \Delta$ we could colour the entire graph with one colour, therefore the assumption $c < n$ and $d < \Delta$ is not restrictive. Additionally, if $d=0$, we know that any two vertices of the same colour do not share an edge, so we cannot have any non-trivial privacy guarantees unless we can use a different colour for each vertex, which is again of no interest, therefore we assume $d \geq 1$.

The bound from Theorem~\ref{thm:lower_bound} in particular says that for constant $\epsilon$ and $c\ll n$, any $\epsilon$-differentially private algorithm which outputs a $(c,d)$-defective colouring with high probability needs $d\in\Omega(\log n)$. Also, note that by the definition of $d$, if we want our algorithm to give a $(c,d)$-defective colouring for any graph of degree $\Delta$, we need $c>\frac{\Delta}{d+1}$ (else, there are graphs such that the colouring fails to be $(c,d)$-defective with probability 1: e.g. the complete graph on $c(d+1)+1$ vertices).


In the following section, we give an algorithm which is gives an $\left(O\left(\frac{\Delta}{\log n}\right), O(\log n)\right)$-defective colouring for constant $\epsilon$ with high probability, which, according to the discussion above, is tight in the sense that i) we cannot hope for an asymptotically lower defectiveness if $c$ is subpolynomial in $n$ and ii) given $d=O(\log n)$, we cannot hope for a palette with asymptotically fewer colors.

\section{An \texorpdfstring{$\epsilon$}{}-edge differentially private 
colouring algorithm}
\label{sec:algorithms}
\begin{theorem}\label{thm:upperbound}
    Let $n, \Delta \in \mathbb{N}$. Let $\epsilon>0$ and $0<\eta\leq 1$.  There is an $\epsilon$-edge differentially private algorithm that, for any input graph $G$ on $n$ vertices, outputs an assignment of colours to the vertices, such that if $G$ has maximum degree $\Delta$, the colour assignment is an $\left(O\left(\eta^2\left(\frac{\Delta}{\log n}+\frac{1}{\epsilon}\right)\right),O\left(\frac{1}{\eta^2}\log n\right)\right)$-defective colouring
    with high probability. In particular, when $\epsilon < 1$, choosing $\eta^2=\epsilon$ gives an $\left(O\left(\frac{\Delta \epsilon}{\log n} + 1\right), O\left(\frac{\log n}{\epsilon}\right)\right)$ colouring.
\end{theorem}

\begin{proof}
Given a graph $G$, we first augment $G$ such that every node in our graph has degree approximately $\Delta$ with high probability. We augment $G$ in the following way: First, we compute an $\epsilon$-differentially private estimate of the maximum degree by $\widetilde{\Delta}=\Delta+ \Lap(1/\epsilon)+\frac{\alpha\log n}{\epsilon}$. Since the $L_1$-sensitivity of the maximum degree is 1, computing $\widetilde{\Delta}$ fulfills $\epsilon$-differential privacy by Lemma~\ref{lem:laplace_mech}. 
By Lemma~\ref{lem:laplace_tail}, the additive error is bounded by $\frac{\alpha\log n}{\epsilon}$ with probability at least $1-n^{-\alpha}$. 
Thus, we get $\Delta + \frac{2\alpha\log n}{\epsilon}\geq \tilde{\Delta}\geq \Delta$ with high probability. We then add $\widetilde{\Delta}$ dummy nodes to our graph and from every $v\in V$, we add $\widetilde{\Delta}$ edges from $v$ to the dummy nodes.  In this augmented graph $\widetilde{G}$, we have $\mathrm{deg}_{\widetilde{G}}(v)=\mathrm{deg}_G(v)+ \tilde{\Delta}$, and thus with probability at least $1-n^{-\alpha}$, every vertex $v$ in $G$ fulfills $\mathrm{deg}_{\widetilde{G}}(v)\in[\widetilde{\Delta}, 2\widetilde{\Delta}]$.

We now colour the resulting graph as follows:
Let $0 < \eta \leq 1$ and $0\leq \beta \leq 1$. We consider a palette $S$ of size:
\[c = |S| = \frac{\eta^2 \widetilde\Delta}{3\log\frac{1}{\beta}}.\]
 Then, all vertices pick a colour uniformly at random from $S$. Our colouring of $G$ is now given by this colouring restricted to the vertices in $G$.
 
 The algorithm fulfills $\epsilon$-differential privacy, since the only time we use information about the edges is to estimate the maximum degree, which we do privately using the Laplace mechanism.

To analyse accuracy, we condition on $\Delta + \frac{2\alpha\log n}{\epsilon}\geq \tilde{\Delta}\geq \Delta$,  which is true with probability $1-n^{-\alpha}$. 
Then for every node $v$ in the original graph $G$,  we analyse the expected number of neighbours in the augmented graph which are assigned the same colour as $v$, which we denote by $d_{G_{conflict}} (v)$. The expectation of $d_{G_{conflict}} (v)$ is given by the degree of $v$ divided by the number of available colours:

\begin{align*}
 E\left(d_{G_{conflict}} (v)\right) &=\frac{\mathrm{deg}_{\widetilde{G}}(v)}{|S|}
 \end{align*}
 We have
 \begin{align*}
    E(d_{G_{conflict}}(v)) =\frac{3\mathrm{deg}_{\widetilde{G}}(v)\log\frac{1}{\beta}}{\widetilde{\Delta}\eta^2} \geq \frac{3\log\frac{1}{\beta}}{\eta^2}
\end{align*}
and 
 \begin{align*}
    E(d_{G_{conflict}}(v)) =\frac{3\mathrm{deg}_{\widetilde{G}}(v)\log\frac{1}{\beta}}{\widetilde{\Delta}\eta^2} \leq \frac{6\log\frac{1}{\beta}}{\eta^2}
\end{align*}
We can show via Chernoff bound that the resulting colouring is $\left(6\frac{(1+\eta)\log(1/\beta)}{\eta^2}\right)$-defective with probability $\beta$. 
Since $\eta\leq 1$, the Chernoff bound Lemma~\ref{lem:chernoff} gives:
\begin{align*}
    P\left(d_{G_{conflict}}(v) \geq (1+\eta) \mu\right) 
    &\leq \exp\left(-\frac{\eta^2}{2 + \eta}\mu\right)\\
    &\leq \exp\left(-\frac{\eta^2}{3}\mu\right)\end{align*}
Plugging in the upper and lower bounds on $\mu$ gives
    \begin{align*}
    P\left(d_{G_{conflict}}(v) \geq \frac{6(1+\eta)\log(1/\beta)}{\eta^2}\right) 
    &\leq \exp\left(-\frac{\eta^2}{3}\frac{3\log\frac{1}{\beta}}{\eta^2}\right)\\
    &\leq \beta
\end{align*}
To get the defectiveness we use a union bound over all vertices. 

\begin{align*}
      P(d < x) &\geq 1 - \sum_v P(deg_{\widetilde{G}}(v) \geq x)
\end{align*}
Plugging in the previous result:
\begin{align*}
   P\left(d < \frac{6(1+\eta)}{\eta^2}\log(1/\beta)\right) &\geq 1 - n\beta
\end{align*}
We want the algorithm to return a valid $(c, d)$-colouring with high probability, that is, we want probability of success $1-n^{\alpha}$ for any constant $\alpha$. We plug in $\beta = 1/n^{\alpha + 1}$:
\begin{align*}
    P\left(d < \frac{6(1+\eta)}{\eta^2}(\alpha + 1)\log(n)\right) &\geq 1 - \frac{n}{n^{\alpha+1}}
\end{align*}

This means that with probability at least $1-n^{-\alpha}$, we have defectiveness
\begin{align*}
    d = O\left(\frac{(1+\eta)}{\eta^2}\log(n)\right) = O\left(\frac{1}{\eta^2}\log(n)\right)
\end{align*}
and
\[c = \Theta\left(\eta^2\left(\frac{\Delta}{\log n} + \frac{1}{\epsilon}\right)\right).\]
Using again union bound, both the condition on the degree and the bound on the defectiveness hold at the same time with probability $1-2n^{-\alpha}$. Choosing $\alpha=\alpha'+\log 2$, we get probability $1-n^{\alpha'}$ for any $\alpha'$.

\end{proof}

With this strategy we will always get $\Omega(\log n)$ defectiveness, even if we only want the result to hold with constant probability, due to the union bound. However, we can increase the defectiveness and, as a trade off, decrease the size of the palette. 

\section{Conclusion}
In this paper, we show that to be $\epsilon$-edge differentially private, a colouring algorithm needs to be defective with defectiveness $d \in \Omega\left( \frac{\log n}{\log c + \epsilon}\right)$ and we propose an $\epsilon$-edge differentially private algorithm using $O\left(\frac{\Delta}{\log n}+\frac{1}{\epsilon}\right)$ colours and defectiveness $O(\log n)$. For any constant $\epsilon$ and defectiveness $\Omega(\log n)$, the number of colours is asymptotically tight, and our lower bound shows that we cannot hope for a much smaller defectiveness unless $c$ or $\Delta$ are very large. In particular, we leave as an open question an algorithm which achieves better defectiveness at the cost of using polynomially many colours in $n$.

 \bibliographystyle{elsarticle-num-names} 
 \bibliography{references.bib}






\end{document}